 \documentclass[a4paper,onefignum,onetabnum]{siamart190516}
\usepackage{amssymb, amsmath}
\usepackage{graphicx}
\usepackage{bm}
 \usepackage{cite}

\def\pra{Phys. Rev. A}

\def\prl{Phys. Rev. Lett.}

\headers{Friedrich-Wintgen bound states in the 
  continuum}{Xiuchen Yu and Ya  Yan Lu} 

\title{Existence of Friedrich-Wintgen bound states in the 
  continuum: system of Schr\"{o}dinger 
  equations\thanks{Submitted to the editors xxx. The main result of
    this paper was presented in 16th International Conference on
    Mathematical and Numerical Aspects of Wave Propagation, June 30 -
    July 5, 2024, with a two-page abstract entitled ``On the existence
    of Friedrich-Wintgen bound states in the continuum.''
\funding{This work was supported by 
the Research Grants Council of Hong Kong
Special Administrative Region, China, under project no.~CityU
11309823.}}}

\author{Xiuchen Yu\thanks{Department of Mathematics, City University
    of Hong Kong, Hong Kong   (\email{yu.xiuchen@my.cityu.edu.hk}).}
\and Ya Yan Lu\thanks{Department of Mathematics, City University of
  Hong Kong, Hong Kong
  (\email{mayylu@cityu.edu.hk}).}}

\begin{document}
\maketitle

\section*{Abstract}
A bound state in the continuum (BIC) is an eigenmode with the
corresponding eigenvalue embedded in the continuous spectrum. 
There is currently a significant research interest on BICs in the
photonics community, because they can be used to induce strong resonances that
are useful for lasing, sensing, harmonic generation, etc.  The existence of BICs
in classical or quantum wave systems has only been established for some
relatively simple cases such as BICs protected by symmetry. In
1985, Friedrich and Wintgen 
(Physical Review A, Vol. 32, pp. 3232-3242, 1985) suggested that
BICs may appear from the destructive interference of two resonances
coupled to a single radiation channel. They used a system of three  
one-dimensional Schr\"{o}dinger equations to illustrate this process.
Many BICs in classical wave systems seem to follow this mechanism and are now called
Friedrich-Wintgen BICs. However, Friedrich and Wintgen did not show
the existence of BICs in their system of three Schr\"{o}dinger equations. Instead, they 
approximated the original system by a model with one
Schr\"{o}dinger equation and two
algebraic equations, and only analyzed BICs in the
approximate model. In this paper, we give a rigorous justification for the
existence of BICs in the original system of three 1D
Schr\"{o}dinger equations.

\section{Introduction}

The concept of bound state in the continuum (BIC) was originally
introduced by von Neumann and Wigner in 1929 for a one-dimensional
(1D) Schr\"{o}dinger equation with an oscillatory
potential~\cite{neumann29}. If the potential in a
Schr\"{o}dinger equation decays to zero at
infinity, an ordinary bound state has a negative energy, while for any 
positive energy, scattering problems with incoming waves from infinity can be
formulated. The bound state constructed by von Neumann and Wigner has
a positive energy, even though the potential decays to zero at
infinity~\cite{still75}. Mathematically, a BIC corresponds to an eigenvalue embedded
in the continuous spectrum, and consequently, the associated scattering problem
loses uniqueness~\cite{fonda63,bonnet94,evans94,mciver96,linton97}. In
addition to Schr\"{o}dinger equations~\cite{fonda63,still75,fried85},
BICs have been found in many classical wave systems governing by the Helmholtz
equation~\cite{bonnet94,evans94,shipman03,porter05,shipman07,lyap15,mai05}, the Maxwell's
equations~\cite{hsu13_2,jin19,yuan21rob}, etc. In recent years, BICs have found many 
important applications in photonics~\cite{Review16,Sadr21,kosh23}.
This is mainly because by perturbing the structure with a BIC  or
selecting a wavevector near the wavevector of the BIC, arbitrarily strong resonances appear~\cite{yuan20,nan25prl}, giving
rise to local field enhancement~\cite{hyl20} and sharp features in
scattering spectra~\cite{shipman12,yzl22} that are useful for
lasing~\cite{kodi17,hwang21}, sensing~\cite{leit19}, harmonic
generation~\cite{kosh19,yuan_siam}, etc.   

There are several different types of BICs. A symmetry mismatch between an
eigenmode and related propagating waves (to or from infinity) gives rise to
symmetry protected BICs~\cite{bonnet94,evans94,shipman07}.
Phase-matched reflections between two scatterers can sometimes
completely trap waves locally forming the so-called Fabry-Perot BICs~\cite{mciver96,linton97,mari08,ches19,mai05}. 
In a seminal work~\cite{fried85}, Friedrich and Wintgen studied a system of three 1D
Schr\"{o}dinger equations, and suggested that BICs may form due to the
destructive interference of two resonances coupled to a single
radiation channel. Since then, many BICs are identified
as the result of destructive interference of two resonances and they are now
called Friedrich-Wintgen BICs (FW-BICs). Typically, a FW-BIC is obtained by
perturbing a system with two bound 
states of which the eigenvalues have a crossing. More
precisely, it is assumed that the two bound states depend on a parameter
$s$, their eigenvalue curves cross at a special parameter value $s_0$, and
the perturbation induces coupling with a radiation channel, so that the two bound states become
resonant states with a complex frequency in general, but at a
particular value of $s$ near $s_0$, one resonant state is actually a bound
state with a real eigenvalue, and it is the FW-BIC. The above scenario
is general and observed in many physical systems~\cite{Review16,lyap15,Sadr21}. 
However, even for the original system of three 1D Schr\"{o}dinger
equations, Friedrich and Wintgen did not provide a justification for
the existence of BICs. In fact,  
they approximated the original system of three Schr\"{o}dinger
equations by a model with
one Schr\"{o}dinger equation and two algebraic equations, and only
analyzed the BICs in the approximate model.

In this paper, we show  that under proper conditions, FW-BICs indeed exist 
in  the system of three 1D Schr\"{o}dinger equations. We
assume that the governing system is a perturbation of an
uncoupled system with two bound states having a crossing in their
eigenvalue curves, and the perturbation induces coupling with a single
radiation channel. In addition, numerical examples are presented to 
support our theory. For comparison, we also show numerical results for
a BIC in the approximate model of  Friedrich and Wintgen. As expected, the BICs in the
original Schr\"{o}dinger system and the approximate model are different. 
The rest of this paper is organized as follows. In section 2, we
recall some basic concepts including BICs, resonant states and
scattering solutions. In section 3, we give a brief summary for
Friedrich and Wintgen's work. Numerical results for both the original 
Schr\"{o}dinger system and the approximate model are presented in
section 4. In section 5, we study boundary value problems for uncoupled
Schr\"{o}dinger equations with given inhomogeneous terms. Section 6
contains our theory on the existence of FW-BICs in the Schr\"{o}dinger system. The paper is concluded
with a brief summary in section 7. 

\section{Basic concepts}

In this section, we recall some basic concepts including bound states,
resonant states, scattering solutions, and BICs.  Our starting point is the 
following system of three 1D Schr\"{o}dinger
equations: 
\begin{equation}
\label{eq:s}
-{\bm y}^{\prime\prime} +V(x) {\bm y} = \lambda {\bm y}, \qquad x>0,
\end{equation}
where $V$ (the potential) is a $3\times 3$ real symmetric matrix
function of $x$, $\lambda$ (the frequency or energy) is a scalar, and 
$ {\bm y} =  \left[ y_0, \ y_1,  \ y_2 \right]^{\sf T}$ is a column
vector depending on $x$. 
We consider the system for $x>0$ only and impose 
the following boundary condition: 
\begin{equation}
  \label{bcat0}
  {\bm y}^\prime(0) = {\bf 0}.
\end{equation}
In addition, we assume there is a constant $L$, such that 
\begin{equation}
  \label{Vinf}
  V(x) =  \mbox{diag}\{ 0, 1, 1 \} = 
  \begin{bmatrix}
    0  && \cr & 1 & \cr && 1
  \end{bmatrix}, \quad x > L.
\end{equation}
More generally, if for sufficiently large $x$, $V$ is an
$x$-independent real symmetric matrix with two equal eigenvalues and
one eigenvalue smaller than these two, we can  always reduce $V$ to the above 
form by redefining ${\bm y}$, shifting $\lambda$,  and scaling $x$.

For Eq.~(\ref{eq:s}) with boundary condition (\ref{bcat0}) and $V$
satisfying Eq.~(\ref{Vinf}), we can 
formulate both boundary-value and eigenvalue problems. 
For $\lambda \in (0, 1)$, we can formulate a scattering problem by
specifying a propagating incident wave from $x=+\infty$ for $y_0$,   that is 
\begin{equation}
\label{def:inc}
   {\bm y}^{(i)}(x) = \begin{bmatrix}
A e^{-{\sf i}\sqrt{\lambda}x} \\
0\\
0 
\end{bmatrix}, \quad x>L, 
\end{equation}
where $A$ is the amplitude of the incident wave. This scattering
problem is a boundary value problem (BVP). For $x> L$, its solution can be written as 
\begin{equation}
  \label{eq:1}
{\bm y}(x) = {\bm y}^{(i)}(x) + A \begin{bmatrix}
  R_0e^{{\sf i} \sqrt{\lambda}x}\\
T_1e^{-\sqrt{1-\lambda} x}\\
T_2e^{-\sqrt{1-\lambda} x}
\end{bmatrix}, 
\end{equation}
where $R_0$ is the reflection coefficient, $T_1$ and $T_2$ are
coefficients of the exponentially decaying $y_1$ and
$y_2$ (evanescent waves). 
It is easy to show that  $|R_0| =1$. 
Let $R_0 = e^{{\sf i}\theta_0}$ and $A =(1/2) e^{- {\sf i}\theta_0/2}$, then  
\begin{equation}
  \label{y0cos}
  y_0(x) = \cos(\sqrt{\lambda} x+\theta_0/2), \quad x > L.
\end{equation}
Since the real
part of ${\bm y}$, i.e., $\mbox{Re}({\bm y})$, solves the same
scattering problem, we can assume ${\bm y}$ is real.
If this scattering problem has a unique solution,
then ${\bm y}$ must be real. But as we will see below, the
scattering problem does not always have a unique solution. 
%  For $\lambda > 1$, the 1st and 2nd channels (for $y_1$ and $y_2$,
%  respective) also support propagating waves, and we can formulate
%  scattering problems with incident propagating waves in all three
%  channels.  

% Define bound state 

If Eqs.~(\ref{eq:s}) and (\ref{bcat0}) are supplemented with the 
boundary condition 
\begin{equation}
\label{bc:s}
{\bm y}(x)  \to {\bf 0}, \quad\text{as}\quad x\to+\infty, 
\end{equation}
we have an eigenvalue problem. A nontrivial solution of (\ref{eq:s}),
(\ref{bcat0}) and (\ref{bc:s}) is a bound state. Normally, a bound state has a
negative energy, i.e., $\lambda <0$. In that case, ${\bm y}$ satisfies 
\begin{equation}
\label{eq:bds}
{\bm y}(x) = \begin{bmatrix}
K_0 e^{ {\sf i}  \sqrt{\lambda}x} \\
K_1e^{ - \sqrt{1-\lambda}x}\\
K_2e^{ - \sqrt{1-\lambda}x} 
\end{bmatrix},
\quad \ x>L, 
\end{equation}
for some constants  $K_0$, $K_1$, $K_2$, and $ {\sf i} \sqrt{\lambda} = -
\sqrt{-\lambda}$. We are interested in bound states with $\lambda >
0$. Such a bound state is a BIC, since it coexists with scattering
solutions that do not decay as $x \to +\infty$. Importantly,  if there is a BIC with 
$\lambda \in (0, 1)$, then Eq.~(\ref{eq:bds}) is still valid, but we must have $K_0=0$, namely $y_0(x) \equiv  0$
for $x > L$. Furthermore, for the scattering problem with the same $\lambda$ as
the BIC, any solution plus a multiple of the BIC is also a solution,
thus, the scattering problem loses uniqueness.
Notice that there is no BIC with $\lambda \ge 1$. If  ${\bm y}$ satisfies (\ref{bc:s}) and $\lambda \ge 1$, then
$y_1$ and $y_2$ must also vanish for all $x > L$. This leads to an initial
value problem for  the system of Schr\"{o}dinger equations with initial
conditions ${\bm y} = {\bm y}' = {\bm 0}$ at
$x=L$. By the uniqueness of the initial value problem, we have ${\bm y}(x)
\equiv {\bm 0}$ for all $x>0$.  

% resonant states

Instead of (\ref{bc:s}), we can look for nontrivial solutions of 
Eq.~(\ref{eq:s}) satisfying  boundary conditions (\ref{bcat0}) and (\ref{eq:bds}).
The bound states 
with $\lambda < 0$ and BICs with $\lambda \in (0, 1)$ 
indeed satisfy (\ref{eq:bds}). For real $\lambda >
0$, condition (\ref{eq:bds}) implies that the wave is outgoing in the 
$y_0$ component.  It is easy to show that there are no solutions
satisfying (\ref{eq:bds}) with $\lambda > 0$ and $K_0 \ne 
0$. However, system (\ref{eq:s}), (\ref{bcat0}) and (\ref{eq:bds}) 
can have solutions with a complex $\lambda$ and a 
nonzero $K_0$, and they are the so-called resonant states. The complex
$\lambda$ of a resonant state should have positive real part and a
negative imaginary part, such that $e^{ {\sf i} \sqrt{\lambda} x}$ is an outgoing
wave with an  increasing amplitude and $e^{ - \sqrt{1-\lambda} x}$ is an
evanescent wave with exponentially decaying amplitude as $x \to
+\infty$. The fact that $\mbox{Im}(\lambda) < 0$ implies that the resonant
state decays with time, as it radiates out energy through the $y_0$
component. A BIC can be regarded as a special resonant state with a real
$\lambda$. 

In the following, we assume $V$ is a perturbation of a
diagonal matrix, namely,
\begin{equation}
  \label{DplusF}
  V(x) = D(x; s) + \delta F(x),
\end{equation}
where $D$ depends on a parameter $s$, $F$ is the perturbation profile
matrix, and $\delta$ is the amplitude of the perturbation. More precisely, 
\begin{equation}
  \label{def:v1}
  D(x;s) = 
\begin{bmatrix}
d_0(x;s)&0&0\\
0&d_1(x;s)&0\\
0&0&d_2(x;s) 
\end{bmatrix},
\quad
F(x) = \begin{bmatrix}
0 & \alpha(x)&\beta(x) \\
\alpha(x)&0 &\gamma(x) \\
 \beta(x)&\gamma(x)&0 
\end{bmatrix}, 
\end{equation}
where $d_0$, $d_1$, $d_2$, $\alpha$, $\beta$ and $\gamma$ are all 
piecewise smooth functions of $x$. In addition, to be consistent with (\ref{Vinf}), we assume 
\begin{equation}
  \label{DFinf}
  D(x; s) \equiv \mbox{diag}\{0, 1, 1\},  \quad F(x) \equiv {\bf 0}, \quad x 
  \ge L.  
\end{equation}

\section{Friedrich and Wintgen's work}

In a pioneering  work~\cite{fried85}, Friedrich and Wintgen studied
Eq.~(\ref{eq:s}) to suggest that BICs can form due to the destructive interference of two
resonances coupled to a single radiation channel. They identified a
condition under which Eq.~(\ref{eq:s}) with $V$ given in (\ref{DplusF}) has a  BIC. The condition is specified on solutions 
of the uncoupled system with $\delta=0$. It is assumed that the uncoupled
equations for $y_1$ and $y_2$ have bound states with frequency in $(0,1)$
for all $s$ in some interval  $I_s$. To avoid confusion, we denote the bound
state corresponding to $y_i$ by $\phi_i$ and its frequency by $\mu_i$ for
$i=1$, 2. Therefore, 
$\phi_i$  and $\mu_i$ satisfy 
\begin{equation}
\label{eq:unp}
\begin{cases}
&-\phi_i^{\prime\prime}+d_i(x;s)\phi_i = \mu_i\phi_i, \quad x > 0, \\
&\phi_i^\prime(0; s) = 0,\\
& \phi_i(x; s) \to 0,  \quad x\to +\infty. 
\end{cases}
\end{equation}
Notice that both $\phi_i$ and $\mu_i$
depend on parameter $s$
 and $\mu_i(s) \in (0, 1)$ for all $s \in I_s$. Friedrich and Wintgen
 suggested that if in the $s$-$\lambda$ plane, the two curves $\lambda = \mu_1(s)$ and $\lambda
 =\mu_2(s)$
 have a crossing at $s_0 \in  I_s$  and $\delta$ is small,  then for
 some $s$  near $s_0$,  Eq.~(\ref{eq:s}) 
 has a BIC with $\lambda$ near
 \begin{equation}
   \label{lam0}
   \lambda_0 :=\mu_1(s_0)=\mu_2(s_0).
 \end{equation}

However, Friedrich and Wingten did not show the existence of BICs for
Eq.~(\ref{eq:s}). Instead, they introduced the following approximate model:
\begin{equation}
  \label{eq:approx}
  \begin{cases}
    &-y_0^{\prime\prime} + d_0y_0 - \lambda y_0 = 
     - \delta \alpha\phi_1 C_1  -  \delta    \beta\phi_2 C_2, \quad x
     > 0 \\
&(\lambda -\mu_1) a_1 C_1   = \delta  a_0    C_2 +  \delta b_1,  \\
&  (\lambda - \mu_2) a_2 C_2 = \delta  a_0   C_1  +  \delta b_2, 
\end{cases}
\end{equation}
where
\[
  a_0  = \langle \phi_1, \gamma \phi_2 \rangle,  \
  b_1 = \langle \phi_1, \alpha y_0 \rangle, \
  b_2 = \langle \phi_2, \beta y_0 \rangle, \
  a_i  = \langle \phi_i, \phi_i \rangle, \  i=1, 2, 
\]
and $\langle \cdot, \cdot\rangle$ denotes the $L^2$ inner product for real
functions on $\mathbb{R}^+$, namely,
\begin{equation}
  \label{L2ip}
  \langle f, g \rangle = \int_0^\infty f(x) g(x) dx.  
\end{equation}
The first equation in (\ref{eq:approx}) corresponds to the original
equation for $y_0$ with $y_1$ and $y_2$ replaced by $C_1 \phi_1$ and $C_2
\phi_2$, respectively. The other two equations in (\ref{eq:approx}) are 
the inner products of $\phi_i$ with the equation of $y_i$, where
$y_1$ and $y_2$ are also replaced by $C_1 \phi_1$ and $C_2 \phi_2$. 
These two equations form a linear system for $C_1$ and $C_2$ which
depend on $y_0$ linearly through the integrals $b_1$ and
$b_2$. 

Together with the boundary conditions $y_0'(0)=0$ and  $y_0(x)
\to 0$ as $x \to +\infty$, Eq.~(\ref{eq:approx}) gives rise to an eigenvalue problem. 
Friedrich and Wingten
showed that a BIC of (\ref{eq:approx}) must
satisfy 
\begin{equation}
  \label{NandD}
{\cal N}(s, \lambda)={\cal D}(s, \lambda) = 0, 
\end{equation}
where ${\cal N}$ and ${\cal D}$ are real functions of
$s$ and $\lambda$,  and they are related to $\phi_1$, $\phi_2$, $F$,
$\delta$, etc.  Friedrich and Wingten argued that if the two
curves $\lambda=\mu_1(s)$ and 
$\lambda=\mu_2(s)$ have a crossing, then (\ref{NandD}) should have a
solution near the crossing point $(s_0, \lambda_0)$. 

Moreover, Friedrich and Wintgen suggested that if the approximate model
(\ref{eq:approx}) has a BIC, then the original system (\ref{eq:s})
should also have a BIC. However, they did not provide any justification. 
Although (\ref{eq:approx}) is derived from (\ref{eq:s}) assuming
$\delta$ is small, there are important differences. As we know from
section 2, a BIC of (\ref{eq:s}) must have $y_0(x) \equiv 0$ for $x >
L$, but for a BIC of the approximate model (\ref{eq:approx}), $y_0$ is
in general nonzero for all $x$, since the first equation in (\ref{eq:approx}) has a nonzero right
hand side. Although Friedrich and Wintgen's work~\cite{fried85} is
very important and widely cited, to the best of our knowledge, the
existence of BICs  for system (\ref{eq:s}) has never been proved.  The
purpose of this paper is to present a rigorous proof.

\section{Numerical examples}
\label{numex}

In this section, we present a numerical example to illustrate the
existence of FW-BICs. We calculate both the exact BIC of the original
system (\ref{eq:s}) and the BIC of the approximate model
(\ref{eq:approx}). The results confirm that they are different. 
The example involves a potential $V$ satisfying
Eqs.~(\ref{DplusF})-(\ref{DFinf}), where $L=4$. The matrices $D$
and $F$ are piecewise constant in $x$ and given in Table~\ref{ex:para}
\begin{table}[htb]
  \centering
  \caption{ 
    Potential $V=D+\delta F$ of a Schr\"{o}dinger system. Elements of
    matrices $D(x; s)$ and $F(x)$ are given in different intervals of  
    $x$.}
  \begin{tabular}[h]{|c||c|c|c|}\hline
    $x$ & $(0,2)$  & $(2,4)$ & $ (4, \infty)$  \\ \hline
    $d_0(x;s)$ & 0 & 0 & 0 \\ \hline
    $d_1(x;s)$ & $-0.2+0.2s$ & $-0.2+0.2s$ & 1 \\ \hline
    $d_2(x;s)$ & $0.4 - 0.5s$ & $-0.6+0.1s$ & 1 \\ \hline
    $\alpha(x)$ & $0.265$ & $-0.805$ & 0 \\ \hline
    $\beta(x)$  & $-0.746$ & $0.827$ & 0 \\ \hline
    $\gamma(x)$ & $0.629$ & $0.812$ & 0 \\ \hline
  \end{tabular}
\vspace{0.2cm}
\label{ex:para}
\end{table}
below. In particular, we let $D$ depend on parameter $s$ linearly.

For this example, we first calculate the two bound states of the
uncoupled system with $V=D$ (i.e. $\delta = 0$). As in the previous section,
the bound state for $y_i$ is denoted as
$\phi_i$ and the corresponding eigenvalue is $\mu_i$, and they 
depend on parameter $s$. In Fig.~\ref{figex}(a),
\begin{figure}[htb]
  \centering
 \includegraphics[width = 6cm]{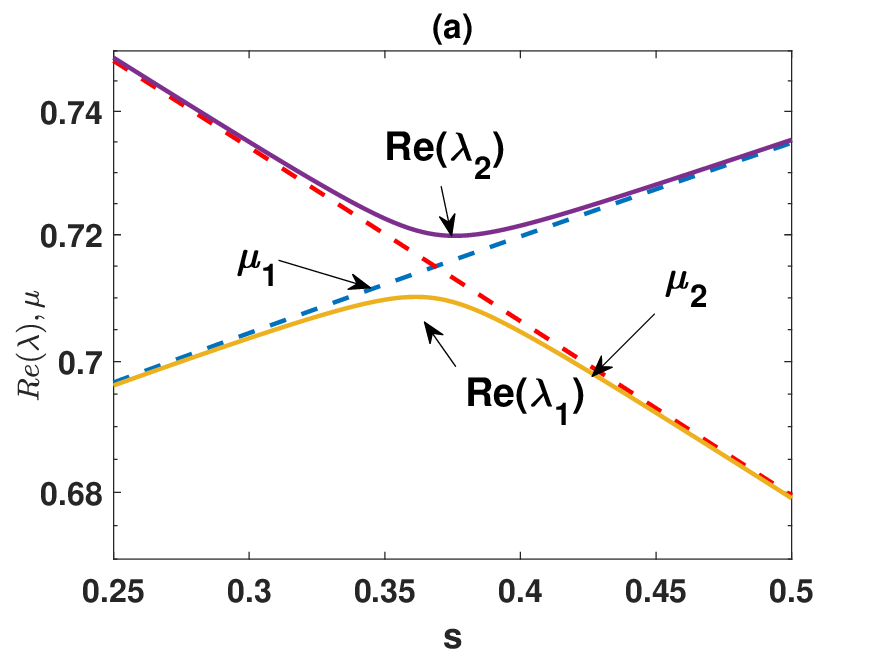}
 \includegraphics[width = 6cm]{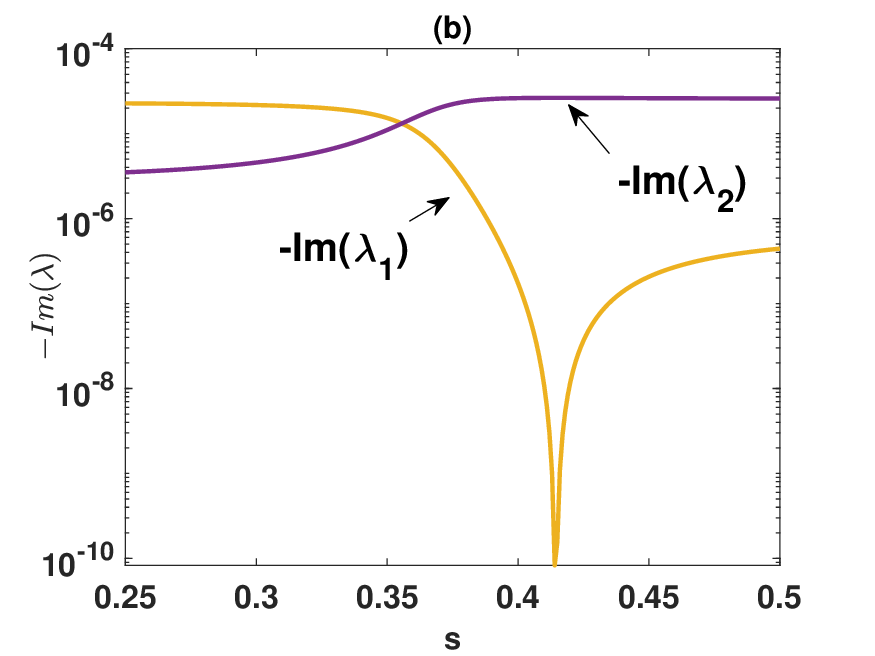}
 \caption{(a)  Eigenvalues $\mu_1$ and $\mu_2$ of two bound states of
   the uncoupled system ($\delta=0$) as functions of parameter $s$, shown as the
   dashed curves. They have a crossing at $s_0 \approx 0.3686$.
   Real parts of $\lambda_1$ and $\lambda_2$ of  
   two resonant states of the coupled system ($\delta = 0.01$) shown 
   as solid curves. (b) 
   Imaginary parts of $\lambda_1$ and $\lambda_2$ of 
   two resonant states of the coupled system ($\delta = 0.01$),
   shown in a logarithmic scale.  A  BIC is found at $s_* \approx
   0.4144$ where $\mbox{Im}(\lambda_1)    = 0$. }
\label{figex}
\end{figure}
we show $\lambda = \mu_1(s)$ and $\lambda = \mu_2(s)$ as dashed curves. Clearly, these two curves have a crossing, and the crossing
point  is $(s_0, \lambda_0) \approx (0.3686,0.7150)$. 
% (0.368601605813080, 0.714952837356155)
Next, we consider the coupled system (\ref{eq:s}) with
$\delta = 0.01$ and calculate resonant states satisfying the boundary conditions
(\ref{bcat0}) and (\ref{eq:bds}). Since the perturbation introduces coupling between the
different components of ${\bm y}$, the two bound states $\phi_1$ and $\phi_2$ of
the uncoupled system become resonant states with complex frequency 
$\lambda_1$ and $\lambda_2$, respectively. In Fig.~\ref{figex}(a), we show the real
parts of $\lambda_1$ and $\lambda_2$ as functions of $s$. Notice that
unlike $\mu_1$ and $\mu_2$, the real parts of $\lambda_1$ and
$\lambda_2$ do not have a crossing. Since the resonant states lose
power through the outgoing wave of $y_0$, the imaginary parts of
$\lambda_1$ and $\lambda_2$ are negative (in general), and they are
shown in Fig.~\ref{figex}(b) in a logarithmic scale. It can be seen
that $\mbox{Im}(\lambda_1)$ reaches zero at $s_* \approx
0.4144$. This corresponds to a FW-BIC of the coupled system
(\ref{eq:s}). The frequency of the BIC is $\lambda_* \approx
0.7012$. Notice that $(s_*, \lambda_*)$ is reasonably close to $(s_0,
\lambda_0)$. 

For the same $\delta=0.01$, the approximate model (\ref{eq:approx})
has a BIC at $s \approx 0.4136$ and the frequency of that BIC is $\lambda
\approx 0.7014$.  Therefore, the BIC of the original system (\ref{eq:s}) is different
from the BIC of the approximate model. It is not clear whether 
the existence of a BIC in the approximate model really implies the
existence of a BIC in original system. In any case, we will show that 
under proper conditions, Eq.~(\ref{eq:s}) indeed has a BIC. 

\section{Uncoupled and inhomogeneous equations}
In order to show BICs  exist in the coupled system (\ref{eq:s}), 
we first  study the uncoupled
system ($V=D$, i.e., $\delta =0$) with inhomogeneous terms. The
results obtained in this section will be used in section~6 to construct
a BIC for the coupled system. 

\subsection{Basic equations}

The uncoupled  and inhomogeneous system at the crossing point is 
\begin{equation}
\label{eq:D0}
\begin{cases}
  &- {\bm y}^{\prime\prime}  +D(x; s_0) {\bm y} - \lambda_0 {\bm y} =
  {\bm g}(x), \qquad x>0,\\
  &{\bm y}^\prime(0) = {\bf 0}, \\
  & {\bm y}(x) \to {\bm 0}, \quad x \to +\infty, 
\end{cases}
\end{equation}
where ${\bm g}(x) =\left[ g_0(x), g_1(x), g_2(x) \right]^{\sf T}$ is a
real vector function in $L^2(\mathbb{R}^+)$ and 
$g_0(x) \equiv 0$ for $x> L$. 
As in section 3, we assume the homogeneous uncoupled system has  two
bound states $\{ \phi_i(x; s), \mu_i(s) \}$ satisfying (\ref{eq:unp})
for $i=1$, 2, and they have a crossing point $(s_0, \lambda_0)$
satisfying $\lambda_0 = \mu_1(s_0) = \mu_2(s_0) \in (0,1)$.
Equation~(\ref{eq:D0}) is defined for the fixed $s=s_0$ and 
$\lambda=\lambda_0$. 
The results obtained in this section will be used 
in a constructive proof for BICs in the coupled system.

The two bound states above are solutions of the uncoupled equations
for $y_1$ and $y_2$, respectively. For $\lambda > 0$, the uncoupled
equation for $y_0$ does not have a bound state. Instead, it has a real scattering solution satisfying 
\begin{equation}
  \label{defphi0}
\begin{cases}
&-{\phi_0^{\prime\prime}}+d_0(x;s ){\phi_0} = \lambda{\phi_0}, \quad x>0,\\
&\phi_0^\prime(0; s, \lambda) = 0,  \\
& \phi_0(x; s, \lambda) =  \cos(\sqrt{\lambda}x+\theta_0/2), \quad x >
L, 
\end{cases}
\end{equation}
where, as discussed in section 2, $\theta_0$ is the phase of the
reflection coefficient $R_0$. For simplicity, we denote $\phi_0(x; s_0, \lambda_0)$,
$\phi_1(x; s_0)$, $\phi_2(x; s_0)$ by $\phi_0(x)$, $\phi_1(x)$ and
$\phi_2(x)$, respectively, and denote $d_i(x; s_0)$ by $d_i(x)$ for
$i=0$, 1, 2.

In the following, for real functions on $\mathbb{R}^+$, we use
$\| \cdot \|$ and $\| \cdot \|_{H^2}$ to
denote their $L^2$-norm and $H^2$-norm,
respectively, and they are defined as follows:
\begin{equation}
  \label{defnorm}
  \| f \|      = \langle f, f \rangle^{1/2}, \quad 
  % \label{h2norm}
  \| f \|_{H^2} = ( \|f \|^2 + \| f' \|^2 + \|f'' \|^2)^{1/2}. 
\end{equation}
In addition, we normalize the two bound states so that 
$\| \phi_1 \| = \| \phi_2 \| = 1$.
% where the $L^2$ inner product is defined in Eq.~(\ref{L2ip}). 

\subsection{BVP for $y_0$}
Since the equations in (\ref{eq:D0}) are uncoupled, we first consider
the BVP for $y_0$, i.e., 
\begin{equation}
\label{BVP:0}
  \begin{cases}
    & -y_0'' + d_0(x) y_0 - \lambda_0 y_0 = g_0(x), \quad x > 0, \\
    & y_0'(0) = 0, \\
    & y_0(x) \to 0, \quad x \to +\infty.
  \end{cases}
\end{equation}
We have the following lemma.
\begin{lemma}
If $g_0(x)$ is a real function in $L^2(\mathbb{R}^+)$, $g_0(x) \equiv
0$ for $x>L$,
$d_0(x)$ is a real piecewise smooth function,  $d_0(x) \equiv 0$ for $x  
> L$, $\lambda_0 > 0$, and $\phi_0$ is the real scattering solution
satisfying {\rm (\ref{defphi0})} with $s=s_0$ and $\lambda=\lambda_0$,
then  {\rm BVP (\ref{BVP:0})} has a solution if and only if $\langle
\phi_0, g_0 \rangle=0$. Moreover, the solution $y_0$ is unique, and there exists a constant $C$ independent of $g_0$, such that 
$\|y_0\|_{H^2} \le C \|g_0\|.$
\end{lemma}
\begin{proof}  For $x > L$, the equation for $y_0$ is simply
$y_0''+ \lambda_0 y_0 = 0$.  Since $\lambda_0 > 0$, if 
BVP (\ref{BVP:0}) has a solution $y_0$ that tends to zero as $x \to
+\infty$, then $y_0(x) \equiv 0$ for
$x > L$.
Multiplying $\phi_0$ to both side of the
inhomogeneous Schr\"{o}dinger equation of $y_0$ and integrating on $(0, \infty)$, we immediately obtain
$\langle \phi_0, g_0 \rangle=0$. Therefore, $\langle \phi_0, g_0 \rangle=0$
is a necessary condition for (\ref{BVP:0}) to have a solution.
The solution is unique, because  if $w$ is the difference of 
two solutions, then it satisfies 
the corresponding homogeneous equation with (initial) 
conditions $w(L)=w'(L)=0$, thus $w(x) \equiv 0$ for all $x > 0$.

Let $w_1$ and $w_2$ be the solutions of the homogeneous
equation with  the following  initial conditions:  
\[
w_1^\prime(0) = w_2(0) = 0, \quad w_1(0) = w_2^\prime(0) = 1.  
\]
It is clear that the scattering solution $\phi_0$ can be written as
$\phi_0(x) = \phi_0(0) w_1(x)$ with $\phi_0(0) \ne 0$. 
The general solution satisfying the inhomogeneous
Schr\"{o}dinger equation and boundary condition $y_0'(0) = 0$ is 
\begin{equation}
\label{eq:gsol}
y_0(x)  = \int_0^x\left[ w_1(t)w_2(x)-w_2(t)w_1(x) \right] g_0(t)\,dt + C_0 w_1(x),
\end{equation}
where $C_0$ is an arbitrary constant.
It is easy to verify that if $\langle \phi_0, g_0 \rangle=0$
and $C_0 = \int_0^L w_2(t) g_0(t) dt$, then $y_0$ given in (\ref{eq:gsol})
satisfies $y_0(x)=y_0'(x)=0$ for all $x \ge L$, and is a  solution
of BVP (\ref{BVP:0}). Therefore, $\langle \phi_0, g_0 \rangle=0$ is
also a sufficient condition.

If $y_0$ is a solution of (\ref{BVP:0}), it is also a solution of a reduced BVP on $(0,L)$
with boundary conditions $y_0'(0)=0$ and $y_0(L)=0$ [or $y_0'(L)=0$].  Since
$y_0(x) \equiv 0$ for $x>L$,  we only need to consider the boundedness
for solutions of  the reduced BVP. Since $d_0$ is piecewise smooth,
according to Ref.~\cite{evans22} (theorem 4 in section 6.3), there
is a constant $C$ independent of  $g_0$, such that
\begin{equation}
\|y_0\|_{H^2}=\|y_0\|_{H^2((0,L))}\le C\|g_0\|_{L^2((0,L))}= C\|g_0\|. 
\end{equation}

\end{proof}

\subsection{BVP for $y_1$}
In the uncoupled system (\ref{eq:D0}),  $y_1$ satisfies the following BVP:
\begin{equation}
\label{BVP:1}
\begin{cases}
& -y_1'' + d_1(x) y_1 - \lambda_0 y_1 = g_1(x), \quad  x>0, \\
& y_1^\prime(0) = 0, \\
& y_1(x) \to  0, \quad x\to+\infty. 
\end{cases}  
\end{equation}
Concerning the solvability and boundedness  
of $y_1$, we have the following lemma. 
\begin{lemma}
  If   $g_1(x)$ is a real function in $L^2(\mathbb{R}^+)$, $d_1(x)$ is a
  piecewise smooth function,  $ d_1(x) \equiv 1$ for $x > L$,
  $\lambda_0 \in (0,1)$,
and $\phi_1$ is a bound state satisfying {\rm (\ref{eq:unp})}, $s=s_0$ and $\mu_1(s_0)=\lambda_0$,   then {\rm BVP (\ref{BVP:1})} is solvable
if and only if  $ \langle \phi_1, g_1 \rangle = 0$, and
  there is a constant $C$ independent of $g_1$, such that 
  $\| y_1^{(p)} \|_{H^2} \le C \| g_1 \|$,
  where $y_1^{(p)}$ is the unique solution of {\rm BVP (\ref{BVP:1})}
  satisfying $\langle \phi_1, y_1^{(p)} \rangle=0$.
\end{lemma}
\begin{proof}
  Since $d_1(x) \equiv 1$ for $x > L$, ${\cal L}_1 := -d^2/dx^2  + d_1(x) -
\lambda_0$, as a linear operator from
$H^2(\mathbb{R}^+)$ to $L^2(\mathbb{R}^+)$, is a compact perturbation of
${\cal L}_1^{\infty} := - d^2/ dx^2  +1 -\lambda_0$.
Therefore,  according to Ref.~\cite{davis07} (lemma
  14.4.1), ${\cal L}_1$ is Fredholm if and only if
${\cal L}_1^\infty$ is Fredholm. 
Since $\lambda_0 < 1$, by Fourier transform,
${\cal L}_1^\infty$ is bijective and thus  Fredholm.
Therefore, ${\cal L}_1$ is also Fredholm. Consequently, 
BVP (\ref{BVP:1}) is solvable if and only if $g_1$ satisfies
$\langle g_1,\phi \rangle  = 0$ for any $\phi \in H^2(\mathbb{R}^+)$
satisfying the homogeneous equation. Since the homogeneous equation has
only one linearly independent bounded solution,  the condition is equivalent to
$\langle g_1, \phi_1 \rangle =0.$

If $y_1$ is a solution of BVP (\ref{BVP:1}), then there is a constant
$C_1$ such that $w = y_1 - C_1 \phi_1$ is also a solution of 
(\ref{BVP:1}), and is orthogonal to $\phi_1$ by the $H^2$ inner
product. In fact, let ${\cal  P}$ be a projection operator of
$H^2(\mathbb{R}^+)$ on the kernel of ${\cal L}_1$, then $C_1 \phi_1 = {\cal P} y_1$. Since ${\cal L}_1 :
H^2(\mathbb{R}^+) \to L^2(\mathbb{R}^+)$ has a closed image and a
one-dimensional kernel,  according to 
Ref.~\cite{volp11} (lemma 3.5 in chapter 6), 
there is a constant $C_2 $ independent of $y_1$, such that
$\| w \|_{H^2}  = \|y_1  - {\cal P} y_1\|_{H^2} \le C_2  \| {\cal L}_1
y_1\| = C_2 \| g_1 \|$.
We are concerned with the particular solution $y_1^{(p)}$ orthogonal
to $\phi_1$ by the $L^2$ inner product. In fact, since 
$\| \phi_1 \| = 1$, we have $y_1^{(p)} = w + C_3 \phi_1$, where $C_3 = - \langle \phi_1,
w\rangle$. Thus, $|C_3 | \le \| \phi_1 \|  \, \| w\|  = \| w\|
\le \| w \|_{H^2}$, and therefore
\[
  \| y_1^{(p)} \|_{H^2} \le
\|w\|_{H^2} + |C_3| \, \| \phi_1\|_{H^2} 
\le (1 + \| \phi_1\|_{H^2})  \|w\|_{H^2} \le C \| g_1 \|,
\]
where $C =  (1 + \| \phi_1\|_{H^2}) C_2$. 
\end{proof}

\subsection{BVP (\ref{eq:D0})}

Since the equations in (\ref{eq:D0}) are uncoupled, and the BVP for
$y_2$ is exactly the same as that for $y_1$, we can derive
solvability conditions and establish the boundedness of the solutions
for BVP (\ref{eq:D0}) by simply putting together the results for each
components of ${\bm y}$. Corresponding to the scattering solution
$\phi_0$ and bound states $\phi_1$ and $\phi_2$, we define the
following vector functions:
\begin{equation}
\label{def:sol}
{\bm \phi}_0  = \begin{bmatrix}
\phi_0\\
0\\
0 
\end{bmatrix}, \quad
{\bm \phi}_1  = \begin{bmatrix}
0\\
\phi_1\\
0 
\end{bmatrix}, \quad
{\bm \phi}_2 = \begin{bmatrix}
0\\
0\\
\phi_2 
\end{bmatrix}. 
\end{equation}
For real column vector functions on $\mathbb{R}^+$, we define the
$L^2$ inner product as
\begin{equation}
\langle {\bm f}, {\bm g} \rangle = \int_0^{\infty} [ {\bm f}(x)]^{\sf
  T} {\bm g}(x) \,dx
\end{equation}
and define the $L^2$-norm and $H^2$-norm as in (\ref{defnorm}).  The
following result is a direct consequence of Lemmas 5.2 and 5.3 
\begin{lemma}
  \label{lem:uncp}
  If ${\bm g}(x) = [ g_0(x), g_1(x), g_2(x) ]^{\sf T}$ is a real vector function in $L^2(\mathbb{R}^+)$,
  $g_0(x) \equiv 0$ for $x>L$, $D(x; s_0) = {\rm diag} \{d_0(x),
  d_1(x), d_2(x) \}$ is a real piecewise smooth
  diagonal matrix function, $D(x; s_0) \equiv  {\rm diag}\{ 0, 1, 1\}$ for $x> L$, $\lambda_0 \in (0,1)$, $\phi_0$ is a real scattering
  solution satisfying {\rm (\ref{defphi0})} for $s=s_0$ and
  $\lambda=\lambda_0$, $\phi_1$ and $\phi_2$ are bound states
  satisfying {\rm (\ref{eq:unp})} for $s=s_0$ and $\mu_1(s_0)=\mu_2(s_0)=\lambda_0$, 
 ${\bm \phi}_0$, ${\bm \phi}_1$, ${\bm \phi}_2$ are vector
  functions given
  in {\rm (\ref{def:sol})}, then {\rm BVP (\ref{eq:D0})} has a solution if and
  only if 
  \begin{equation}
    \label{eq:solcond}
    \langle {\bm \phi}_i, {\bm g} \rangle  = 0,\quad i=0, 1, 2, 
  \end{equation}
  and there is a constant $C$ independent of ${\bm g}$, such that 
  \begin{equation}
\label{eq:bdd}
\| {\bm y}^{(p)} \|_{H^2} \le  C \|{\bm g}\|,
\end{equation}
where ${\bm y}^{(p)}$ is a particular solution satisfying $\langle {\bm
  \phi}_i, {\bm y}^{(p)} \rangle  = 0$ for $i=1, 2$. 
\end{lemma}

\section{Friedrich-Wintgen BICs}
In this section, we present a constructive proof for the existence of
a BIC in the coupled system. The BIC and the parameter $s$ are
constructed as power series of $\delta$, and we show that the series
converge if $\delta$ sufficiently small  in absolute value. 

\subsection{Assumptions}

So far, we have assumed that the diagonal matrix $D(x; s)$ depends a 
parameter $s$, and the uncoupled system has two bound states with a
crossing at $s_0$. To prove the existence of FW-BICs in the coupled
system, we need an additional assumption on the perturbation matrix
$F$. For simplicity, we further assume that $D$ depends on $s$ linearly.
The assumptions are summarized as follows.
\begin{enumerate}
\item [A1:]  The diagonal matrix $D$ depends on $s$ linearly and is
  written as
  \begin{equation}
  \label{Dlinear}
  D(x;s) = D_0(x)+(s-s_0)D_1(x),   \quad D_0(x) := D(x;s_0).
\end{equation}
\item [A2:] The uncoupled system, (\ref{eq:s}) with $\delta=0$ or
  $V(x)=D(x;s)$, has 
  two bound states $\{ \phi_i(x; s),  \mu_i(s) \}$ 
  satisfying
  (\ref{eq:unp}) and $\mu_i(s) \in (0,1)$ for $s \in I_s$ and $i \in
  \{1, 2\}$. These two bound states have a crossing at $s_0 \in I_s$, such
  that 
  \begin{equation}
    \label{nontan}
    \lambda_0 := \mu_1(s_0) = \mu_2(s_0), \quad 
    \langle {\bm \phi}_1, D_1 {\bm \phi}_1 \rangle \ne 
    \langle {\bm \phi}_2, D_1 {\bm \phi}_2 \rangle,
  \end{equation}
  where ${\bm \phi}_1$ and ${\bm \phi}_2$ are given in
  Eq.~(\ref{def:sol}) and normalized 
  such that
  \begin{equation}
    \label{normphi12}
    \| {\bm \phi}_1 \| = \| {\bm \phi}_2 \| = 1.    
  \end{equation}
\item [A3:] The perturbation matrix $F$ satisfies
  \begin{equation}
  \label{coupling}
  \langle {\bm \phi}_0, F {\bm \phi}_j \rangle \ne 0, \quad j=1, 2, 
\end{equation}
where ${\bm \phi}_0$ is given in Eq.~(\ref{def:sol})  and normalized
such that
\begin{equation}
  \label{normphi0}
  | \langle {\bm \phi}_0, F {\bm \phi}_1 \rangle |^2 
  +   | \langle {\bm \phi}_0, F {\bm \phi}_2 \rangle |^2 = 1.  
\end{equation}
\end{enumerate}

Notice that in the $s$-$\lambda$ plane, the two curves
$\lambda = \mu_1(s)$ and $\lambda=\mu_2(s)$ should not be tangential
to each other at $(s_0, \lambda_0)$. This leads to the second part of
(\ref{nontan}). 
When $\delta$ in $V(x) = D(x; s) + \delta F(x)$ is nonzero, the 
equations in (\ref{eq:s}) are supposed to be coupled. Condition
(\ref{coupling}) ensures that the bound states $\phi_i$
($i=1$ and  2) are coupled to the radiation channel (i.e., the scattering solution
$\phi_0$), so that they become resonant states in the coupled
system.  To simplify the notations, we let 
\begin{equation}
  \label{dijfij}
  d_{ij} = \langle {\bm \phi}_i,D_1{\bm \phi}_j \rangle,\quad f_{ij} =  \langle 
  {\bm \phi}_i,F{\bm \phi}_j \rangle, \quad \mbox{for} \quad i, j = 0, 1, 2.
\end{equation}

\subsection{Power series}

For the coupled system (\ref{eq:s}) with a small nonzero $\delta$, we look
for a BIC $\{ {\bm y}, \lambda \}$ which we assume can be expanded in  power
series of $\delta$.  Since the BIC is expected at some $s$ near
$s_0$, we also expand $s$ in a power series of
$\delta$. Therefore, 
\begin{equation}
  \label{eq:series}
  {\bm y} =\sum_{n=0}^{\infty}{\bm y}_n\delta^n, \quad
  \lambda   =\sum_{n=0}^{\infty}\lambda_n \delta^n,\quad
  s=\sum_{n=0}^{\infty}s_n\delta^n. 
\end{equation}
Inserting the above series into Eq.~(\ref{eq:s}) and collecting terms
at different orders of $\delta$, we obtain the following 
equations: 
\begin{equation}
  \label{eq:ordern}
  - {\bm y}_n'' + D_0(x) {\bm y}_n - \lambda_0 {\bm y}_n 
  = {\bm g}_n(x),    \quad x > 0,
\end{equation}
where ${\bm  g}_0(x) = {\bm 0}$ and 
\begin{equation}
  \label{rhsgj}  
  {\bm g}_n(x) =    \sum_{i=1}^n \left[ \lambda_i - s_i D_1(x) \right] {\bm 
    y}_{n-i}  - F(x)  {\bm y}_{n-1}, \quad n \ge 1.   
\end{equation}
The boundary conditions are 
\begin{equation}
  \label{bcyj}
  {\bm y}_n^\prime(0) = {\bf 0}, \quad  {\bm y}_n(x) \to {\bm 0},\
  x\to+\infty, \quad n\ge 0. 
\end{equation}
For (\ref{eq:ordern})-(\ref{bcyj}), we show that a nontrivial ${\bm
  y}_0$ can be determined, and for each $n\ge 1$, 
$\lambda_n$, $s_n$, ${\bm y}_n$ can be solved and they are real. We
have the follow theorem.
\begin{theorem}
  If assumptions {\rm A1}, {\rm A2} and {\rm A3} are valid, then
  \begin{enumerate}
  \item ${\bm y}_0$ is a real nonzero linear combination of ${\bm \phi}_1$ and 
${\bm \phi}_2$;
\item For each $n\ge 1$, $\lambda_n$, $s_n$  and another number
  $\gamma_n$ can be  determined and they are real;
\item  For each $n\ge 1$, Eq.~(\ref{eq:ordern}) and boundary condition
  (\ref{bcyj}) have a real solution 
\begin{equation}
  \label{solyj}
  {\bm y}_n = {\bm y}_n^{(p)} + \gamma_n {\bm y}_0^\perp, 
\end{equation}
where ${\bm y}_n^{(p)}$ is a real particular 
solution orthogonal to ${\bm \phi}_1$ and ${\bm \phi}_2$, ${\bm 
  y}_0^\perp $ is a real linear combination of ${\bm 
  \phi}_1$ and ${\bm \phi}_2$  and it is orthogonal to ${\bm y}_0$.
\end{enumerate}
\end{theorem}

\begin{proof}
Using the power series of $s$, the potential of $V$ in  (\ref{eq:s}) can
be written as 
\[
  V = D_0+(s-s_0)D_1+\delta F
  = D_0 + (s_1 D_1 + F) \delta + s_2 D_1 \delta^2 + s_3 D_1 \delta^3 +
  \dots
\]
Inserting the power series of ${\bm y}$, $\lambda$, and $V$ above
into Eq.~(\ref{eq:s}), we obtain Eq.~(\ref{eq:ordern}) from the  ${\cal
  O}(\delta^n)$ terms.  If ${\bm y}_0$ is real and satisfies (\ref{eq:ordern})
and (\ref{bcyj}) for $n=0$, then 
\begin{equation}
{\bm y}_0 = C_1{\bm \phi}_1+C_2{\bm \phi}_2, 
\end{equation}
where $C_1$ and $C_2$ are real constants.
From Lemma 5.3, we know that Eq.~(\ref{eq:ordern}) has a solution satisfying
boundary condition (\ref{bcyj}) if and only if its right hand side
${\bm g}_n$ satisfies 
$\langle {\bm \phi}_i,
{\bm g}_n \rangle = 0$ for $i=0, 1, 2$.

For $n=1$, we have ${\bm g}_1 = (\lambda_1  - s_1 D_1 - F) {\bm
  y}_0$. The condition $\langle {\bm \phi}_0, {\bm g}_1 \rangle = 0$
leads to 
\[
  C_1 f_{01} + C_2 f_{02} = 0.
\]
Since the BIC, if it exists, can be scaled by an arbitrary constant, we can choose the following solution
\begin{equation}
  \label{solC1C2}
  C_1 = f_{02}, \quad C_2 = - f_{01}.  
\end{equation}
Due to the scaling of ${\bm \phi}_0$ in assumption A3, we have 
$C_1^2+C_2^2 = 1$. Therefore,  $\| {\bm y}_0 \| = 1$.
The other two conditions, $\langle {\bm \phi}_i, {\bm g}_1 \rangle =
0$ for $i=1, 2$, give rise to the following system
\begin{equation}
  \label{lams1}
  A  \begin{bmatrix}
\lambda_1\\
s_1 
\end{bmatrix}
=   \begin{bmatrix}
  f_{11} & f_{12}\\
  f_{21} &  f_{22}
\end{bmatrix}
\begin{bmatrix} 
C_1 \\
C_2 
\end{bmatrix},
\end{equation}
where
\begin{equation}
  \label{detmatA}
  A =
  \begin{bmatrix}
    C_1 & 0 \cr
    0 & C_2 
  \end{bmatrix}
\begin{bmatrix} 
  1 & -d_{11}\\
  1 &- d_{22}
\end{bmatrix}.  
\end{equation}
Assumptions A2 and A3 ensure that matrix $A$ is real and
invertible, thus, $\lambda_1$ and $s_1$ can be solved and they are
real. With the above $\lambda_1$ and $s_1$, Eq.~(\ref{eq:ordern}) is
solvable, but the solution is not unique. If ${\bm y}_1$ is a
solution, ${\bm y}_1 $ plus any linear combination of ${\bm \phi}_1$
and ${\bm \phi}_2$ is also a solution. Therefore, we can find a
particular solution ${\bm y}_1^{(p)}$ that is orthogonal to both ${\bm
  \phi}_1$ and ${\bm \phi}_2$. Any solution orthogonal to ${\bm
  y}_0$ can be written as 
\begin{equation}
  \label{eq:y1}
{\bm y}_1 =  {\bm y}_1^{(p)}+\gamma_1{\bm y}_0^\perp, 
\quad \mbox{where} \quad 
{\bm y}_0^\perp = C_2{\bm \phi}_1 -C_1{\bm \phi}_2 
\end{equation}
and $\gamma_1$ is some constant.

For $n\ge 2$, we already have ${\bm y}_{n-1} = {\bm y}_{n-1}^{(p)} +
\gamma_{n-1} {\bm y}_0^\perp$,   where $\gamma_{n-1}$ is still
undetermined.
If we rewrite  the right hand side of Eq.~(\ref{eq:ordern}) as 
$ {\bm g}_n =   ( \lambda_n - s_n D_1) {\bm y}_0 -  {\bm h}_n$, where
\begin{equation}
  \label{defhn}
  {\bm h}_n =   F {\bm     y}_{n-1} +   \sum_{i=1}^{n-1} ( s_i D_1 - \lambda_i ) 
{\bm y}_{n-i}, 
\end{equation}
and define
\begin{equation}
  \label{hjp}
  {\bm h}_n^{(p)}  :=
  F {\bm    y}_{n-1}^{(p)}+   \sum_{i=1}^{n-1} ( s_i D_1 - \lambda_i ) 
  {\bm y}_{n-i}^{(p)},   
\end{equation}
then the first solvability condition
$\langle {\bm \phi}_0, {\bm g}_n \rangle = 0$ gives rise to 
\begin{equation}
\label{eq:gamman}
\gamma_{n-1}= \langle {\bm \phi}_0,  {\bm h}_n^{(p)} \rangle.
\end{equation}
The other two conditions, $\langle {\bm \phi}_i, {\bm 
  g}_n \rangle = 0$ for $i=1$ and $2$, lead to 
\begin{equation}
\label{eq:coeffn}
A \begin{bmatrix}
\lambda_n \\
s_n
\end{bmatrix}
=
\begin{bmatrix}
\langle {\bm \phi}_1,  {\bm h}_n  \rangle
\\
\langle {\bm \phi}_2,  {\bm h}_n  \rangle
\end{bmatrix}. 
\end{equation}

Notice that Eq.~(\ref{eq:coeffn}) is also valid for the case of $n=1$
where ${\bm h}_1 = F {\bm   y}_0$. Since all three conditions on ${\bm
  g}_n$ are satisfied, Eq.~(\ref{eq:ordern}) with boundary condition
(\ref{bcyj}) is solvable. The solution ${\bm y}_n$ orthogonal to ${\bm y}_0$ is
given in Eq.~(\ref{solyj}), where $\gamma_n = \langle {\bm \phi}_0,
{\bm h}_{n+1}^{(p)} \rangle$, and ${\bm h}_{n+1}^{(p)}$ is
related to ${\bm y}_n^{(p)}$, ${\bm y}_{n-1}^{(p)}$,  ... , ${\bm y}_1^{(p)}$. 
\end{proof}

\subsection{Convergence}

In this subsection, we show that for sufficiently small $\delta$ (in
absolute value), all power series in (\ref{eq:series}) converge. For that
purpose, we first establish some results about the norms of  the solution ${\bm
  y}_n$, the right hand side
${\bm   g}_n$, and the related ${\bm h}_n$ defined in (\ref{defhn}).
\begin{lemma}
  \label{lem:3}
  If assumptions {\rm A1}, {\rm A2} and {\rm A3} are true, then
  for any $n\ge 1$,   the 
right hand side ${\bm g}_n$ and the solution ${\bm y}_n$ of
Eq.~(\ref{eq:ordern}), and  ${\bm h}_n$ defined in (\ref{defhn}),  satisfy  
\begin{eqnarray}
\label{eq:hiter}
 &&  \|{\bm g}_n\| \le  M_1\|{\bm h}_n\|, \\
&&    \|{\bm h}_n\| \le 
   M_2\left(\|{\bm y}_{n-1}\|+\sum_{i=2}^{n-1}\|{\bm y}_{n-i}\|\cdot\|{\bm h}_i\|\right), \\
& & \|{\bm y}_n\|_{H^2}  \le 
   M_3 \left(\|{\bm h}_n\|+\sum_{i=2}^{n-1}\|{\bm h}_i\|\cdot\|{\bm h}_{n+1-i}\|\right),  
\end{eqnarray}
where $M_1$, $M_2$, and $M_3$ are constants independent  of $n$.   
\end{lemma}

\begin{proof}
  Let $B = [ b_{ij}(x)]$ be an $N \times N$ matrix function of $x$, 
  the maximum norm of $B$ is defined as
\[
\|B \|_{\rm max} = \max_{1\le i, j\le N}  \|b_{ij}(x)\|_\infty. 
\]
From Eq.~(\ref{eq:coeffn}) and using condition (\ref{normphi12}), we deduce that 
\begin{equation}
  \label{lams}
\max\{|\lambda_n|,|s_n|\}\le 2\|A^{-1}\|_{\rm max} \, \|{\bm h}_n\|,\quad  n\ge 1.  
\end{equation}
This leads to 
\[
\| \lambda_n I - s_n D_1 \|_{\rm max}\le M_0\|{\bm h}_n\|,  \quad n\ge 1,
\]
where $I$ is the $3 \times 3 $ identity matrix and 
$M_0 = 2(1+\|D_1\|_{\rm max}) \, \|A^{-1}\|_{\rm max}$.
Since $\lambda_n I - s_n D_1$ is a diagonal matrix and $\| {\bm y}_0
\| = 1$, we have  
\[
  \| {\bm g}_n \| \le
  \| ( \lambda_n - s_n D_1) {\bm y}_0 \| + \| {\bm h}_n \|
  \le \| \lambda_n I - s_n D_1 \|_{\rm max} + \| {\bm h}_n \| \le M_1 \| {\bm
    h}_n \|,
\]
where $M_1 = 1 + M_0$.

As given in Eq.~(\ref{defhn}), the first term
of ${\bm h}_n$ is $F {\bm  y}_{n-1}$. Since $F$ has zero diagonals, we
can show that 
\[
    \| F {\bm y}_{n-1} \| \le 2 \| F \|_{\rm max} \,  \| {\bm y}_{n-1}  \|.
\]
However, ${\bm y}_0$ is a unit vector with only two nonzero
components, and we obtain 
\begin{equation}
\label{h1xxx}
\| {\bm h}_1 \| = \| F {\bm y}_0 \| \le
\sqrt{3} \| F \|_{\rm max} \,  \| {\bm   y}_0 \| =
\sqrt{3} \| F \|_{\rm max}.
\end{equation}
For $n\ge 2$, we have 
\begin{eqnarray*}
 \| {\bm h}_n\| &\le & \|F {\bm y}_{n-1}\|+\sum_{i=1}^{n-1}\|  \lambda_i 
I - s_i D_1 \|_{\rm max}  \, \| {\bm y}_{n-i}\| \\
&  \le & 2 \| F \|_{\rm max} \, \| {\bm y}_{n-1}  \| + M_0 \sum_{i=1}^{n-1}
  \| {\bm y}_{n-i} \|  \cdot \| {\bm h}_i \| \\
  &\le &   M_2 \left(  \| {\bm y}_{n-1}\|+\sum_{i=2}^{n-1}\| {\bm
         y}_{n-i}\| \cdot \|  {\bm h}_i\| \right).  
\end{eqnarray*}
where  $M_2 = \max \{ (2 + \sqrt{3} M_0) \| F\|_{\rm max},  M_0 \}$.

To estimate ${\bm y}_n$, we first rewrite $\gamma_n$ as 
\[
  \gamma_n = \langle {\bm \phi}_0, {\bm h}_{n+1}^{(p)} \rangle
  =
  \langle {\bm \phi}_0,   (F + s_1 D_1 - \lambda_1) {\bm 
    y}_{n}^{(p)} \rangle
  + \sum_{i=2}^n \langle {\bm \phi}_0, (s_i D_1 - \lambda_i) {\bm
    y}_{n+1-i}^{(p)} \rangle.
\]
Since $F(x)$ and the first component of ${\bm y}_n^{(p)}$ vanish for
all $x > L$, we can show that
\begin{equation}
  \label{getK1}
|   \langle {\bm \phi}_0,   (F + s_1 D_1 - \lambda_1) {\bm 
  y}_{n}^{(p)} \rangle   |
\le K_1 \| {\bm y}_n^{(p)} \|_{L},  
\end{equation}
where $\| \cdot \|_L$ denotes the $L^2$ norm on $(0,L)$, i.e., 
\[
 \| {\bm  f} \|_L = \left \{ \int_0^L [ {\bm f}(x)]^{\sf T} {\bm f}(x)
   \, dx  \right \}^{1/2}, 
\]
and $K_1 = (\sqrt{2}+ \sqrt{3} M_0) \| {\bm \phi}_0 \|_L \, \|
F\|_{\rm max}$.  Therefore, 
\begin{eqnarray*}
  |\gamma_n | &\le &  K_1\| {\bm y}_n^{(p)} \| + \sum_{i=2}^{n}
                     \| {\bm \phi}_0 \|_L \,
                     \| \lambda_i I - s_i D_1\|_{\rm max} \,
                     \| {\bm y}_{n+1-i}^{(p)}\| \\
              & \le &  K_2
                      \left(  \|{\bm y}_n^{(p)} \|+\sum_{i=2}^{n}\|{\bm 
          h}_i\|\cdot\|{\bm y}_{n+1-i}^{(p)}\|\right), 
\end{eqnarray*}
where $K_2 = \max \{ K_1, M_0 \| {\bm \phi}_0 \|_L \}.$
Consequently,
\begin{eqnarray*}
  \|{\bm y}_n\|_{H^2}  &\le &
                       \|{\bm y}_n^{(p)} \|_{H^2} +|\gamma_n|\cdot\|{\bm y}_0^\perp\|_{H^2}  \\
&\le &   \|{\bm y}_n^{(p)}\|_{H^2} + K_2 \|{\bm y}_0^\perp\|_{H^2}
                      \left(  \|{\bm y}_n^{(p)} \|+\sum_{i=2}^{n}\|{\bm 
          h}_i\|\cdot\|{\bm y}_{n+1-i}^{(p)}\|\right), \\
&\le  &
K_3 \left(  \|{\bm y}_n^{(p)} \|_{H^2}+\sum_{i=2}^{n}\|{\bm 
        h}_i\|\cdot\|{\bm y}_{n+1-i}^{(p)}\|\right),
\end{eqnarray*}
where $K_3 =  1+ K_2 \|{\bm y}_0^\perp\|_{H^2}$. 
Using Lemma 5.3, (\ref{eq:hiter}) and (\ref{h1xxx}), we obtain
\begin{eqnarray*}
  \|{\bm y}_n\|_{H^2}    
% & \le & \left(  1+ K_2 \|{\bm y}_0^\perp\|_{H^2} \right) C_g 
% \left(  \| {\bm g}_n\|+\sum_{i=2}^{n}\|{\bm h}_i\|\cdot\|{\bm g}_{n+1-i} \|\right), \\
& \le &  
K_3 C_g M_1 
 \left(  \| {\bm h}_n\|+\sum_{i=2}^{n}\|{\bm h}_i\|\cdot \| {\bm h}_{n+1-i} \|\right), \\
  &=& 
K_3 C_g M_1 
       \left(  \| {\bm h}_n\|+ \| {\bm h}_n \| \cdot 
       \| {\bm h}_1 \| +
       \sum_{i=2}^{n-1}\|{\bm h}_i\|\cdot  \|  {\bm h}_{n+1-i} \| \right), \\                    
&\le &  M_3 \left(  \| {\bm h}_n\|+\sum_{i=2}^{n-1} \| {\bm h}_i \|
       \cdot  \| {\bm h}_{n+1-i}\|\right), 
\end{eqnarray*}
where $C_g$ is the constant $C$ in (\ref{eq:bdd}) and 
$ M_3=    K_3 C_gM_1 \left( 1+\sqrt{3}\|F\|_{\rm  max} \right)$. 
\end{proof}

Next, we establish bounds for ${\bm y}_n$ and ${\bm h}_n$ using a
single constant $M$ and an integer sequence $(P_k)_{k=0}^\infty$.
\begin{lemma}
  \label{lem:4}
  If assumptions {\rm A1}, {\rm A2} and {\rm A3} are true, then
  the solution ${\bm y}_n$ of Eq.~(\ref{eq:ordern}) and ${\bm h}_n$
  defined in (\ref{defhn}) satisfy 
  \begin{eqnarray}
\label{eq:iter}
& \|{\bm h}_n\|\le M^{2n-2}P_{n-1}, & \quad  n\ge 2, \\
\label{ynP}
    & \|{\bm y}_n\|_{H^2}\le M^{2n-1}P_{n-1}, & \quad n\ge 1,    
\end{eqnarray}
where $M = \max\{ M_2,M_3 \}$, $(P_k)_{k=0}^\infty$ is an integer sequence
defined recursively by
\[
P_0 = 1,\  P_1 = 1,\ P_k = P_{k-1}  +\sum_{i=1}^{k-1}P_{i} (P_{k-i}  +P_{k-i-1}). 
\]
\end{lemma}

\begin{proof}
  We prove by induction. For $n=1$, since
  ${\bm y}_1 = {\bm y}_1^{(p)} + \gamma_1 {\bm y}_0^\perp$ and 
  \[
    \gamma_1 = \langle {\bm \phi}_0, {\bm h}_2^{(p)} \rangle
  = \langle {\bm \phi}_0, (F + s_1 D_1 -\lambda_1) {\bm y}_1^{(p)}
  \rangle,
\]
we have, as a special case of (\ref{getK1}),
$  |\gamma_1| \le K_1 \| {\bm y}_1^{(p)} \|_L \le K_1 \| {\bm
  y}_1^{(p)} \|_{H^2}$.
Therefore
\[
\| {\bm y}_1\|_{H^2} \le \|{\bm y}_1^p\|_{H^2}+|\gamma_1|\cdot\|{{\bm y}_0^\perp}\|_{H^2}
\le \left( 1 + K_1 \|{{\bm y}_0^\perp}\|_{H^2} \right)
\| {\bm y}_1^{(p)} \|_{H^2}.
\]
By Lemma 5.3 and denoting the constant by $C_g$, we have $\| {\bm
  y}_1^{(p)} \|_{H^2} \le C_g \| {\bm   g}_1 \|$, where
${\bm g}_1 = (\lambda_1 - s_1 D_1 - F) {\bm y}_0$. As shown in the proof of
Lemma 6.2,  $  \| {\bm g}_1 \| \le M_1 \|{\bm h}_1 \| \le \sqrt{3} \| F\|_{\rm
  max} M_1$. Therefore,
\[
  \| {\bm y}_1\|_{H^2} \le \sqrt{3} \| F\|_{\rm max} \, M_1 C_g  \left( 1 +
    K_1 \|{{\bm y}_0^\perp}\|_{H^2} \right) < M_3 \le M P_0, 
\]
where $P_0=1$. For $n= 2$,  by Lemma 6.2, we have
\[
  \| {\bm h}_2\| \le M_2 \| {\bm y}_1\| \le M^2 P_1,
  \quad \| {\bm   y}_2\|_{H^2} \le M_3 \| {\bm h}_2\| \le M^3 P_1, 
\]
where $P_1=1$. 

Now suppose (\ref{eq:iter}) and (\ref{ynP}) are valid for all $n\le k$,
  then for $n=k+1$, we have
\[
 \|{\bm h}_{k+1}\| \le 
   M_2 \left(\|{\bm y}_{k}\|+\sum_{i=2}^k\|{\bm y}_{k+1-i}\|\cdot \|{\bm h}_i\|\right) 
   \le M^{2k}\left(P_{k-1}+\sum_{i=2}^kP_{k-i}P_{i-1}\right).    
 \]
For ${\bm y}_{k+1}$, we use the above and get 
\begin{eqnarray*}
 \|{\bm y}_{k+1}\|_{H^2} &\le &  M_3
   \left(\|{\bm h}_{k+1}\|+\sum_{i=2}^{k}\|{\bm h}_i\|\cdot\|{\bm h}_{k+2-i}\|\right)\\
&\le&   M^{2k+1}\left[
      P_{k-1}+\sum_{i=2}^{k}P_{i-1}(P_{k+1-i}+P_{k-i})\right]
      = M^{2k+1} P_k. 
  \end{eqnarray*}
Since the right hand side of the inequality for $\| {\bm h}_{k+1}\|$
can also be bounded by $M^{2k}P_k$, both (\ref{eq:iter}) 
  and (\ref{ynP})  are valid for $n = k+1$.  
\end{proof}

Finally, we establish the convergence of the power series
(\ref{eq:series}). 
\begin{lemma}
  If assumptions {\rm A1}, {\rm A2} and {\rm A3} are true, then for
  any sufficiently small $\delta$ (in absolute value), all power series in
  (\ref{eq:series}) converge. 
\end{lemma}
\begin{proof} It is clear that $(P_n)_{n=0}^\infty$ is an increasing sequence. Therefore,
  \[
    P_n = P_{n-1}+\sum_{i=1}^{n-1}P_{i}(P_{n-i}+P_{n-i-1})\le
    3\sum_{i=1}^{n-1}P_{i}P_{n-i}.
  \]
  If we define a sequence $( Q_n )_{n=0}^\infty$ by $Q_0=Q_1=1$ and
  \[
    Q_n = 3   \sum_{i=1}^{n-1}Q_{i}Q_{n-i}, \quad n\ge 2,
  \]
  then $P_n \le Q_n$ for all $n\ge 0$. It is easy to see that for
  $n\ge 1$, $Q_n = 3^{n-1} q_{n-1}$, where $q_n$ is Catalan number given by 
\[
q_n = \frac{1}{n+1}\binom{2n}{n},\quad q_n\sim \frac{4^n}{n^{3/2}\sqrt{\pi}}.  
\]
If $|\delta| < 1/(12M^2)$,  
\[
\sum_{n=k+1}^{\infty}\|{\bm y}_n\|_{H^2}\delta^n  \le
\sum_{n=k+1}^{\infty}\frac{M^{2n-1}Q_{n-1} }{12^n M^{2n}} \le
\frac{1}{9M} \sum_{n=k+1}^{\infty} \frac{ q_{n-2}}{4^n}\to 0,\quad  k
\to+\infty, 
\]
thus the power series of ${\bm y}$ in (\ref{eq:series})
converges. For the same $\delta$, due to inequalities (\ref{lams}) and (\ref{eq:iter}),
the power series for $\lambda$ and $s$ also converge. 
\end{proof}

\subsection{Main result}
We summarize our result as follows. 
\begin{theorem}
\label{thm:main}
Consider the system of three  Schr\"{o}dinger equations (\ref{eq:s}) with
boundary condition (\ref{bcat0}) and a matrix potential $V(x) = D(x;
s) + \delta F(x)$, where $s$ is a parameter, $D$ is a diagonal matrix
function, $F$ is a real symmetric matrix function with zero diagonals,
$\delta$ is the amplitude of the perturbation, $D$ and $F$ satisfy
Eqs.~(\ref{def:v1}) and (\ref{DFinf}) for some 
$L > 0$. If assumptions {\rm A1}, {\rm A2} and {\rm A3} are 
true, then for any sufficiently small $\delta$ (in absolute value), 
the system has a BIC. 
\end{theorem}

Recall that assumption A1 states that $D$ depends on $s$
linearly, A2 gives the crossing point $(s_0, \lambda_0)$ of two bound
stats of the uncoupled system (i.e. $\delta=0$), and A3 ensures 
that $F$ induces coupling between the bound states and the 
scattering solution of the uncoupled system. The BIC of the coupled
system ($\delta \ne 0$ and small) is found for $(s, \lambda)$ near
$(s_0, \lambda_0)$. Theorem 6.5 is a direct
consequence of Theorem 6.1 and Lemmas 6.2--6.4. The BIC and parameter
$s$ are constructed as power series of $\delta$ with details given in the
proof of Theorem 6.1, and  according to Lemma 6.4,  these power series converge when $|\delta|$
is sufficiently small.

\section{Conclusion}
\label{con}

The FW-BICs are an important class of BICs resulting from the
destructive interference of two resonances, and they have been found
in many classical and quantum wave systems~\cite{Review16,Sadr21,kosh23}. 
Friedrich and Wintgen proposed this formation mechanism using a system
of three 1D Schr\"{o}dinger equations~\cite{fried85}, but they only 
analyzed an approximation model and did not show the existence of BICs
in the original system. In this paper, we gave a constructive proof
for the existence of FW-BICs in that Schr\"{o}dinger
system. We assumed that the potential $V$ is a
perturbation of a diagonal matrix function, and clarified the
conditions under which FW-BICs exist. 

Although there are numerous works on BICs, especially in the photonics
literature, their existence has only been established for a few simple
cases, such as the symmetry-protected BICs. The destructive
interference interpretation for FW-BICs is intuitive and useful, but
it does not provide a rigorous justification for their existence. 
FW-BICs appear in many classical wave systems. Further studies are
needed to understand their mathematical properties.

% \bibliographystyle{plain}
% \bibliography{ref.bib}

\begin{thebibliography}{99}

\bibitem{neumann29}  J. von Neumann and E. Wigner, {\it \"{u}ber
  merkw\"{u}rdige diskrete Eigenwerte}, Phys. Z,  30 (1929),
pp. 465--467. 


\bibitem{still75} F. H. Stillinger and D. R. Herrick, {\it Bound 
    states in the continuum}, \pra, 11 (1975), pp. 446--454. 


\bibitem{fonda63} L. Fonda, {\it Bound states embedded in the 
    continuum and the formal theory of scattering}, Annals of Physics, 
  22 (1963), pp. 123-132.
  
\bibitem{bonnet94} A.-S. Bonnet-Bendhia and F. Starling, {\it Guided 
  waves by electromagnetic gratings and nonuniqueness examples for the 
  diffraction problem}, Math. Methods Appl. Sci., 17 (1994),
pp. 305--338. 

\bibitem{evans94}  D. V. Evans, M. Levitin and  D. Vassiliev, 
   {\it Existence theorems for trapped modes}, 
   J. Fluid Mech., 261 (1994), pp. 21--31.
   
\bibitem{mciver96} M. McIver, {\it An example of non-uniqueness in the
    two-dimensional water-wave problem},  J. Fluid Mech., 315 (1996),
  pp. 257--266. 

\bibitem{linton97} C. M. Linton and N. Kuznetsov,  {\it Non-uniqueness
    in  two-dimensional water wave problems: numerical evidence and
    geometrical restrictions},   Proc. R. Soc. Lond. A,  453  (1997), 
  pp. 2437--2460. 


  
\bibitem{fried85} H. Friedrich and D. Wintgen,  {\it Interfering 
  resonances and bound states in the continuum}, \pra,
32 (1985), pp. 3231-3242.

 \bibitem{shipman03} S. P. Shipman and S. Venakides, 
  {\it Resonance and bound states in photonic crystal slabs}, 
  SIAM J. Appl. Math., 64 (2003), pp. 322--342.

  \bibitem{porter05} R. Porter and D. Evans,  {\it Embedded Rayleigh-Bloch 
  surface waves along periodic rectangular arrays}, 
Wave Motion, 43 (2005), pp. 29--50. 

\bibitem{shipman07} S. Shipman and D. Volkov,  {\it Guided modes in 
  periodic slabs: existence and nonexistence}, SIAM J. Appl. Math., 67
(2007), pp. 687--713.

\bibitem{lyap15} A. A. Lyapina, D. N. Maksimov, A. S. Pilipchuk, and
  A. F. Sadreev, {\it Bound states in the continuum in open acoustic
    resonators},  Journal of Fluid Mechanics, 780 (2015), pp. 370--387.

\bibitem{mai05} Z. Mai and Y. Y. Lu, {\it Relationship between total
    reflection and Fabry-Perot bound states in the continuum}, 
  \pra, 111 (2025), 013527.


\bibitem{hsu13_2} C. W. Hsu, B. Zhen, J. Lee, S.-L. Chua,
  S. G. Johnson, J. D. Joannopoulos, and M. Solja\v{c}i\'{c},
  {\it Observation of trapped light within the radiation continuum},
  Nature, 499 (2013), pp. 188--191.

\bibitem{jin19} J. Jin, X. Yin, L. Ni, M. Solja\v{c}i\'{c}, B. Zhen,
  and C. Peng, {\it Topologically enabled ultrahigh-$Q$ guided
    resonances robust to out-of-plane scattering},  Nature, 
  574  (2019), pp. 501--505.

\bibitem{yuan21rob}  L. Yuan and Y. Y. Lu,
  {\it Conditional robustness of propagating bound states in the
    continuum in structures with two-dimensional periodicity},
  \pra, 103 (2021), 043507.

 \bibitem{Review16} C. W. Hsu, B. Zhen, A. D. Stone,
  J. D. Joannopoulos, and M. Solja\v{c}i\'{c}, 
{\it Bound states in the  continuum},  Nat. Rev. Mater., 1 (2016),
16048. 

\bibitem{Sadr21} A. F. Sadreev, {\it Interference traps waves in an open
    system: bound states in the continuum},   Rep. Prog. Phys.,  84 (2021), 055901. 

\bibitem{kosh23}   K. L. Koshelev, Z. F. Sadrieva, A. A. Shcherbakov, Y. S. 
Kivshar, and A. A. Bogdanov, {\it Bound states in the continuum in
  photonic structures}, Phys.-Uspekhi, 93 (2023), pp. 528. 

\bibitem{yuan20} L. Yuan and Y. Y. Lu,  {\it Perturbation theories for
symmetry-protected bound states in the continuum on
two-dimensional periodic structures}, \pra, 101 (2020), 043827.

\bibitem{nan25prl} N. Zhang and Y. Y. Lu, {\it Perturbation Theory for
    Resonant States near a Bound State in the Continuum}, \prl, 134
  (2025), 013803.

\bibitem{hyl20} Z.  Hu, L. Yuan, and Y. Y. Lu, 
{\it Resonant field enhancement near bound states in the continuum on
  periodic structures}, \pra,  101 (2020), 043825. 

\bibitem{shipman12} S. P. Shipman and H. Tu, {\it Total resonant transmission
and reflection by periodic structures}, SIAM J. Appl.
Math.,  72 (2012), pp. 216--239.

\bibitem{yzl22} L. Yuan, M. Zhang, and Y. Y. Lu,
  {\it Real transmission and reflection zeros of periodic structures
    with a bound state in the continuum}, \pra,  106 (2022), 013505.

  %% applications
\bibitem{kodi17} A. Kodigala, T. Lepetit, Q. Gu, B. Bahari, Y. Fainman, and
B. Kant\'{e},   {\it Lasing action from photonic bound states in
  continuum}, Nature,  541 (2017), pp. 196--199.

\bibitem{hwang21} M.-S. Hwang, H.-C. Lee, K.-H. Kim, K.-Y. Jeong, S.-H.
Kwon, K. Koshelev, Y. Kivshar, and H.-G. Park, {\it Ultralow-threshold
  laser using super-bound states in the continuum}, 
Nat. Commun. 12 (2021), 4135.

\bibitem{leit19} A. Leitis, A. Tittl, M. Liu, B. H. Lee, M. B. Gu, Y. S.
Kivshar, and H. Altug,   {\it Angle-multiplexed all-dielectric
metasurfaces for broadband molecular fingerprint retrieval}, 
Sci. Adv. 5 (2019), eaaw2871.

\bibitem{kosh19} K. Koshelev,  Y.  Tang, K. Li, D.-Y. Choi. G. Li, and
  Y. Kivshar, {\it Nonlinear metasurfaces governed by bound states in
    the continuum}, ACS Photonics, 6 (2019), pp.~1639--1644.

\bibitem{yuan_siam} L. Yuan and Y. Y. Lu, {\it Excitation of bound
    states in the continuum via second harmonic generations}, SIAM
  J. Appl. Math.,  80 (2020), pp~864--880.
  
\bibitem{mari08} D. C. Marinica, A. G. Borisov, and S. V. Shabanov,
  {\it Bound states in the continuum in photonics}, \prl,   100 (2008), 183902. 

\bibitem{ches19} L. Chesnel and V. Pagneux, {\it From zero transmission to
trapped modes in waveguides}, J. Phys. A: Math. Theor.,  52 (2019),
165304.

\bibitem{evans22} L. C. Evans, Partial Differential Equations, 2nd
  edition, American Mathematical Society, 2022.

\bibitem{davis07} E. B. Davis, Linear Operators and their Spectra,
  Cambridge University Press, 2007.

\bibitem{volp11} V. Volpertm, Elliptic Partial Differential
  Equations, Volume 1: Fredholm Theory of Elliptic Problems in
  Unbounded Domains, Springer Science \& Bisiness Media, 2011.

\end{thebibliography}
\end{document}